\documentclass{llncs}
\usepackage{amsthm}

\usepackage[utf8]{inputenc}
\usepackage[dvipsnames]{xcolor}
\usepackage{fullpage}
\usepackage[english]{babel}
\usepackage{amssymb,mathtools}
\usepackage{xspace}
\usepackage{bm}
\usepackage[pdfpagelabels]{hyperref}
\usepackage{csquotes}
\usepackage{color}
\usepackage{microtype}
\usepackage{nicefrac}
\usepackage{tasks}
\usepackage[inline]{enumitem}
\usepackage{comment}
\usepackage{subcaption}
\usepackage{tikz}
\usepackage[sort&compress,numbers]{natbib}

\makeatletter
\newtheorem*{rep@theorem}{\rep@title}
\newcommand{\newreptheorem}[2]{%
\newenvironment{rep#1}[1]{%
 \def\rep@title{#2 \ref{##1}}%
 \begin{rep@theorem}}%
 {\end{rep@theorem}}}
\makeatother

\newreptheorem{theorem}{Theorem}
\newreptheorem{lemma}{Lemma}
\newcommand{\Oh}{\mathcal{O}}
\newcommand{\OhOp}[1]{\Oh\mathopen{}\mathclose\bgroup\left( #1 \aftergroup\egroup\right)}
\DeclareMathOperator{\poly}{{\rm poly}}

\usepackage{xspace}
\newcommand{\FPT}{{\sf FPT}\xspace}

\newcommand{\NPh}{\hbox{{\sf NP}-hard}\xspace}

\usepackage{tabularx}

\usepackage[textwidth=20mm,obeyFinal,colorinlistoftodos]{todonotes}

\DeclareMathOperator{\cl}{\mathrm{cl}}
\DeclareMathOperator{\height}{\mathrm{height}}
\def \transpose {{\intercal}}
\DeclareMathOperator{\td}{\mathrm{td}}
\DeclareMathOperator{\ttd}{\mathrm{th}}
\DeclareMathOperator{\fr}{\mathrm{frac}}
\DeclareMathOperator{\lcm}{\operatorname{lcm}}
\newcommand{\CC}{\mathcal{C}}

\newcommand{\III}{\mathcal{I}}

\newcommand{\probfont}[1]{\textsc{#1}}

\usepackage{stackengine}
\stackMath

\newcommand{\df}{:=}

\usetikzlibrary{fit,matrix,patterns,decorations.pathreplacing}

\newcommand*\matbox[6][]{\node [label=center:#6, inner sep=-2pt, fit= (m-#2-#3) (m-#4-#5)] [draw=#1] {};}
\newcommand*\blockmat[3]{\small
\begin{tikzpicture}[baseline=(current bounding box.base)]
	\matrix (m) [ampersand replacement=\&, matrix of math nodes,left delimiter=(,right delimiter=),minimum width=2em,minimum height=2em,nodes in empty cells]
	{
	\&\&\&\&\\
	\&\&\&\&\\
	\&\&\&\&\\
	\&\&\&\&\\
	\&\&\&\&\\
	};
	\matbox{1}{1}{2}{2}{#1}
	\matbox{3}{1}{5}{1}{#2}
	\matbox{3}{3}{5}{5}{#3}
\end{tikzpicture}
}
\def\ve#1{\mathchoice{\mbox{\boldmath$\displaystyle\bf#1$}}
{\mbox{\boldmath$\textstyle\bf#1$}}
{\mbox{\boldmath$\scriptstyle\bf#1$}}
{\mbox{\boldmath$\scriptscriptstyle\bf#1$}}}

\newcommand\veb{{\ve b}}
\newcommand\vecc{{\ve c}}

\newcommand\vel{{\ve l}}

\newcommand\veu{{\ve u}}

\newcommand\vex{{\ve x}}
\newcommand\vey{{\ve y}}
\newcommand\vez{{\ve z}}

\newcommand\vezero{{\ve 0}}

\def\R{\mathbb{R}}
\def\Z{\mathbb{Z}}
\def\N{\mathbb{N}}

\def\Q{\mathbb{Q}}

\usepackage{etoolbox}

\newcommand{\ILPF}{\probfont{ILP-feasibility}}
\newcommand{\MILPF}{\probfont{MILP-feasibility}}
\newcommand{\MILP}{\probfont{MILP}}

\begin{document}
\title{Parameterized Algorithms for MILPs with Small Treedepth}

\author{Cornelius Brand\inst{1} \and
Martin Koutecký\inst{1} \and
Sebastian Ordyniak\inst{2}}
\institute{Computer Science Institute, Charles University, Prague, Czech Republic \and Department of Computer Science, University of Sheffield, United Kingdom}

\maketitle

\begin{abstract}
Solving (mixed) integer linear programs, (M)ILPs for short, is a fundamental optimization task.
While hard in general, recent years have brought about vast progress for solving structurally restricted, (non-mixed) ILPs: $n$-fold, tree-fold, 2-stage stochastic and multi-stage stochastic programs admit efficient algorithms,
and all of these special cases are subsumed by the class of ILPs of small treedepth.

In this paper, we extend this line of work to the mixed case, by showing an
algorithm solving MILP in time $f(a,d) \poly(n)$, where $a$ is the largest coefficient of the constraint matrix, $d$ is its treedepth, and $n$ is the number of variables.

This is enabled by proving bounds on the denominators of the vertices of bounded-treedepth (non-integer) linear programs. We do so by carefully analyzing the inverses of invertible submatrices of the constraint matrix. This allows us to afford scaling up the mixed program to the integer grid, and applying the known methods for integer programs.

We trace the limiting boundary of our approach, showing that naturally related classes of linear programs have vertices of unbounded fractionality. 
Finally, we show that restricting the structure of only the integral variables in the constraint matrix does not yield tractable special cases.
\end{abstract}

\medskip
\noindent

\sloppy

\section{Introduction}
%\label{sec:introduction}
Integer linear programming is a fundamental hard problem, which motivates the search for tractable special cases.
In the '80s, Lenstra and Kannan~\cite{Kannan:1987,Lenstra:1983} and Papadimitriou~\cite{Papadimitriou81} have shown that the classes of ILPs with few variables or few constraints and small coefficients, respectively, are polynomially solvable.
A line of research going back almost 20 years~\cite{HOR,ChenM18,EisHK18,AH,HemmeckeKW:2014,GanianOrdyniakRamanujan17,GanianOrdyniak18,DvorakEGKO17} has recently culminated with the discovery of another tractable class of ILPs~\cite{EisenbrandHunkenschroederKleinKouteckyLevinOnn19,KouteckyLO18}, namely ILPs with small treedepth and coefficients.
The obtained results already found various algorithmic applications in areas such as
scheduling~\cite{KnopK18,CMYZ17,JansenKMR19},
stringology and social
choice~\cite{KnopKM17,KnopKM17b}, and the traveling salesman problem~\cite{ChenM18}.

The language of ``special tractable cases'' has been developed in the theory of parameterized complexity~\cite{CyganFKLMPPS:2015}.
We say that a problem is \emph{fixed-parameter tractable} (\FPT) parameterized by $k$ if it has an algorithm solving every instance $I$ in time $f(k) \poly(|I|)$ for some computable function $f$, and we call this an \emph{\FPT algorithm}.
Say that the height of a rooted forest is its largest root-leaf distance.
A graph $G=(V,E)$ has treepdepth $d$ if $d$ is the smallest height of a rooted forest $F=(V,E')$ in which each edge of $G$ is between an ancestor-descendant pair in $F$, and we write $\td(G)=d$.
The \emph{primal graph} $G_P(A)$ of a matrix $A \in \R^{m \times n}$ has a vertex for each column of $A$, and two vertices are connected if an index $k \in [m] = \{1, \dots, m\}$ exists such that both columns are non-zero in row $k$.
The \emph{dual graph} $G_D(A)$ is defined as $G_D(A) \df G_P(A^{\intercal})$.
Define the \emph{primal treedepth of $A$} to be $\td_P(A) = \td(G_P(A))$, and analogously $\td_D(A) = \td(G_D(A))$.
The recent results states that there is an algorithm solving ILP in time $f(\|A\|_\infty, \min\{\td_P(A), \td_D(A)\}) \poly(n)$, hence ILP is \FPT parameterized by $\|A\|_\infty$ and $\min\{\td_P(A), \td_D(A)\}$.
Besides this class, other parameterizations of ILP have been successfully employed to show tractability results, such as bounding the treewidth of the
primal graph and the largest variable domain~\cite{JansenKratsch15}, the treewidth of the incidence graph and the largest solution prefix sum~\cite{GanianOrdyniakRamanujan17}, or the signed clique-width of the incidence graph~\cite{EibenGKO18}.

It is therefore natural to ask whether these tractability results can
be generalized to more general settings than ILP.
In this paper we ask this question for Mixed ILP (MILP), where both integer and non-integer variables are allowed:
\begin{equation}\label{MILP}
\min \left\{ \vecc\vex \mid A\vex = \veb, \, \vel \leq \vex \leq
\veu \,, \vex \in \Z^{z}\times \Q^{q} \right\} \,, \tag{\MILP{}}
\end{equation}
with $A \in \Z^{m \times z+q}$, $\vel, \veu, \vecc \in \Z^{z+q}$ and $\veb \in \Z^{m}$.

MILP is a prominent modeling tool widely used in practice.
For example, Bixby~\cite{Bixby:2002} says in his famous analysis of LP solver speed-ups, \emph{``[I]nteger programming, and most particularly the mixed-integer variant, is the dominant application of linear programming in practice.''}
Already Lenstra has shown that MILP with few integer variables is polynomially solvable, naturally extending his result on ILPs with few variables.
Analogously, we seek to extend the recent tractability results from ILP to MILP,
most importantly for the parameterization by treedepth and largest coefficient.
Our main result is as follows:
\begin{theorem} \label{thm:main}
MILP is \FPT parameterized by $\|A\|_\infty$ and $\min\{\td_P(A), \td_D(A)\}$.
\end{theorem}
We note that our result also extends to the inequality form of MILP with constraints of the form $A\vex \leq \veb$ by the fact that introducing slack variables does not increase treedepth too much~\cite[Lemma 56]{EisenbrandHunkenschroederKleinKouteckyLevinOnn19}.

The proof goes by reducing an MILP instance to an ILP instance whose parameters do not increase too much, and then applying the existing algorithms for ILP.
A key technical result concerns the \emph{fractionality} of an MILP instance, which is the minimum of the maxima of the denominators in optimal solutions.
For example, it is well-known that the natural LP for the \textsc{Vertex Cover} problem has half-integral optima, that is, there exists an optimum with all values in $\{0, \frac{1}{2}, 1\}$.

The usual way to go about proving fractionality bounds is via Cramer's rule and a sufficiently good bound on the determinant. As witnessed by any proper integer multiple of the identity, determinants can grow large even for matrices of very benign structure. Instead, we need to analyze much more carefully the structure of the inverse of the appearing invertible submatrices, allowing us to show:
\begin{theorem} \label{thm:frac}
A MILP instance with a constraint matrix $A$ has an optimal solution $\vex$ whose largest denominator is bounded by $g(\|A\|_\infty, \min\{\td_P(A), \td_D(A)\})$ for some computable function $g$.
\end{theorem}
We are not aware of any prior work which lifts a positive result for ILP to a result for MILP in this way.

\medskip

Let us relate Theorem~\ref{thm:main} to the well-studied classes of $2$-stage stochastic and $n$-fold ILP.
These ILPs have a constraint matrix composed of small blocks whose largest number of columns and rows, respectively for $2$-stage stochastic and $n$-fold ILP, is denoted $t$, and whose largest coefficient is $a=\|A\|_\infty$.
The bound given by Theorem~\ref{thm:frac} is then $a^{t^{\Oh(t^2)}}$ (cf. Remark~\ref{rem:nfoldbound}).
Consequently, the algorithmic results we obtain by using the currently state-of-the-art algorithms~\cite{EisenbrandHunkenschroederKleinKouteckyLevinOnn19} are near-linear \FPT algorithms parameterized by $a$ and $t$.
\begin{corollary} \label{cor:milp}
Let $g := {a^{a^{t^{\Oh(t^2)}}}}$.
$2$-stage stochastic MILP is solvable in time $2^g \cdot n \log^3 n \log \|\veu-\vel\|_\infty \log \|\vecc\|_\infty$, and
$n$-fold MILP is solvable in time $g \cdot n \log n \log \|\veu-\vel\|_\infty \log \|\vecc\|_\infty$.
\end{corollary}
While solving the $2$-stage stochastic and $n$-fold LP is possible in polynomial time, the best known dependence of a general algorithm for LP on the dimension $n$ is superquadratic.
Eisenbrand et al.~\cite[Theorem 63]{EisenbrandHunkenschroederKleinKouteckyLevinOnn19} give parameterized algorithms whose dependence on $n$ is near-linear, but which depend on the required accuracy $\epsilon$ with a term of $\log(1/\epsilon)$.
General bounds on the the encoding length of vertices of polyhedra~\cite[Lemma 6.2.4]{GLS} imply that to obtain a vertex solution, we need to set $\epsilon$ to be at least as small as $\frac{1}{(an)^n}$, hence $\log(1/\epsilon) \geq (n \log an)$, making the resulting runtime superquadratic.
Since the bound of Theorem~\ref{thm:frac} does not depend on $n$, we in particular obtain first near-linear \FPT algorithms for LPs with small coefficients and small $\td_P(A)$ or $\td_D(A)$.
We spell out the resulting complexities for the aforementioned important classes:
\begin{corollary} \label{cor:lp}
	$2$-stage stochastic LP is solvable in time $2^{2a^{\Oh(t^3)}} \cdot n \log^3 n \log \|\veu-\vel\|_\infty \log \|\vecc\|_\infty$, and
	$n$-fold LP is solvable in time $a^{t^{\Oh(t^2)}} \cdot n \log n \log \|\veu-\vel\|_\infty \log \|\vecc\|_\infty$.
\end{corollary}

\medskip

We also explore the limits of approaching the problem by bounding the fractionality of inverses:
Other ILP classes with parameterized algorithms involve constraint matrices with small primal treewidth~\cite{JansenKratsch15} and $4$-block $n$-fold matrices~\cite{HemmeckeKW:2014}.
Here, we obtain a negative answer:
\begin{lemma} \label{lem:highfracmilp}
For every $n \in \N$, there are MILP instances $I_1$ and $I_2$ with constraint matrices $A_1$ and $A_2$, such that $A_1$ has constant primal, dual, and incidence treewidth and $\|A_1\|_\infty = 2$, and $A_2$ is 4-block $n$-fold with all blocks being just $(1)$, and the fractionality is $2^{\Omega(n)}$ for $I_1$ and $\Omega(n)$ for $I_2$.
\end{lemma}
Next, we consider extending the positive result of Theorem~\ref{thm:main} to separable convex functions, which is the regime considered in~\cite{EisenbrandHunkenschroederKleinKouteckyLevinOnn19}.
We show that merely bounding the fractionality will unfortunately not suffice:
\begin{lemma} \label{lem:highfracmip}
There are MIP instances with the following properties:
\begin{enumerate}
	\item $A=(1 \cdots 1)$, $b=1$, $f(\vex) = \sum_i (x_i)^2$, $\td_D(A)=1$, fractionality $n$,
	\item dimension $1$, no constraints, $f(x) = (x - \frac{1}{k})^2$, fractionality $k$,
	\item dimension $1$, no equality constraints, $0 \leq x \leq 1$, $f(x) = x^3 + 2x^2 - x$ univariate cubic convex, unbounded fractionality (minimum is $\frac{\sqrt{7}}{3} - \frac{2}{3}$).
\end{enumerate}
\end{lemma}

Finally, we consider a different way to extend tractable ILP classes to MILP.
Divide the constraint matrix $A$ of an MILP instance in two parts corresponding to the integer and continuous variables as $A=(A_{\Z}~A_{\Q})$.
What structural restrictions have to be placed on $A_{\Z}$ and $A_{\Q}$ in order to obtain tractability of MILP?
We show a general hardness result in this direction:
\begin{lemma} \label{lem:hard}
  Let $\CC$ be a class of ILP instances for which the feasibility decision problem is \NPh.
  Then there exists a class of MILP instances $\CC'$ whose feasibility decision problem is \NPh and
  whose constraint matrix is $A=\left(\begin{array}{cc}\vezero & A_\Q \\ I &-I\end{array}\right)$,
  where $I$ is the identity matrix and $A_{\Q}$ is a
  constraint matrix of an instance from $\CC$.
\end{lemma}
Note that the main reason for intractability is that we allow
arbitrary interactions between the integer and the non-integer
variables of the instance.
Thus, Lemma~\ref{lem:hard} implies that this interaction between integral and fractional variables has to be restricted in some way in order to obtain a tractable fragment of MILP.

\subsection*{Related Work}
We have already mentioned related work on structural parameterizations of ILP.
The closest work to ours was done by Hemmecke~\cite{Hemmecke:2000} in 2000 when he studied a mixed-integer test set related to the Graver basis, which is the engine behind all recent progress on ILPs of small treedepth.
It is unclear how to apply his approach, however, because it requires bounding the norm of elements of the mixed-integer test set, where the bound obtained by (a strenghtening of)~\cite[Lemma 6.2]{Hemmecke:2000},\cite[Lemma 2.7.2]{Hemmecke:diss},
%(attributed to an unpublished manuscript of Foroudi and Graver)
is polynomial in $n$, too much to obtain an \FPT algorithm.
Kotnyek~\cite{Kotnyek:2002} characterized $k$-integral matrices, i.e., matrices whose solutions have fractionality bounded by $k$, however it is unclear how his characterization could be used to show Theorem~\ref{thm:frac}, so we take a different route.
Lenstra~\cite{Lenstra:1983} showed how to solve MILPs with few integer variables using the fact that a projection of a polytope is again a polytope; applying this approach to our case would require us to show that if $P$ is a polytope described by inequalities with small treedepth, then a projection of $P$ also has an inequality description of small treedepth.
This is unclear.
Half-integrality of two-commodity flow~\cite{Hu:1963,Karzanov:1998} and \textsc{Vertex Cover}~\cite{NemhauserT:1974} has been known for half a century.
Ideas related to half-integrality have recently led to improved \FPT algorithms~\cite{IwataWY:2016,IwataYY:18,Guillemot:2011}, some of which have been experimentally evaluated~\cite{PilipczukZ:2018}.

\section{Preliminaries}
We consider zero a natural number, i.e., $0 \in \N$.
We write vectors in boldface (e.g., $\vex, \vey$) and their entries in normal font (e.g., the $i$-th entry of~$\vex$ is~$x_i$).
For positive integers $m \leq n$ we set $[m,n] \df \{m,\ldots, n\}$ and $[n] \df [1,n]$.

\subsection{Reducing MILP to ILP}
Assume that an MILP instance is given and that some optimum $\vex = (\vex_{\Z}, \vex_{\Q})$ exists whose set of denominators is $D$, and we know $M = \max D$.
Recall $\lcm(D)$ is the least common multiple of the elements of $D$, and $\lcm(D) \leq M! \eqqcolon \tilde{M}$.
Then $\lcm(D) \vex_{\Q}$ is an integral vector.
Our idea here is to restrict our search among all optima of~\eqref{MILP} to search among those optima with small fractionality, that is, with small denominators.
Consider the \emph{integralized MILP} instance:
\begin{equation} \label{IMILP}
\min \{\vecc \vez \colon (\tilde{M} \cdot A_{\Z}~A_{\Q}) \vez = \tilde{M} \cdot \veb, (\vel_{\Z}, \tilde{M} \vel_{\Q}) \leq (\vez_{\Z}, \vez_{\Q}) \leq (\veu_{\Z}, \tilde{M} \veu_{\Q}), \, \vez \in \Z^{z+q}\} \enspace .\tag{IMILP}
\end{equation}
We claim that the optimum of~\eqref{MILP} can be recovered from the optimum of~\eqref{IMILP}:
\begin{lemma} \label{lem:IMILP}
Let $M$ be the fractionality of~\eqref{MILP} and $(\vez_{\Z}~\vez_{\Q}) \in \Z^{z+q}$ be an optimum of~\eqref{IMILP}.
Then $\vex = (\vez_{\Z}~\frac{1}{\tilde{M}} \vez_{\Q})$ is an optimum of~\eqref{MILP}.
\end{lemma}
\begin{proof}
It is clear that there is a bijection between solutions $\vex$ of~\eqref{MILP} where $\vex_{\Q}$ has all entries with a denominator $\tilde{M}$ and solutions $\vez$ of~\eqref{IMILP}.
The optimality of $\vex$ then follows from $M$ being the fractionality of~\eqref{MILP} and $M!$ always being divisible by $\lcm(D)$.
\end{proof}

\subsection{The Graphs of $A$ and Treedepth}

\begin{comment}
\begin{definition}[Primal and dual graph]\label{primaldual-graph}
	Given a matrix $A \in \Z^{m \times n}$, its \emph{primal graph} $G_P(A) = (V,E)$ is defined as $V = [n]$ and $E = \left\{\{i,j\} \in \binom{[n]}{2} ~\middle|~ \exists k \in [m]: A_{k,i}, A_{k,j} \neq 0\right\}$.
	In other words, its vertices are the columns of $A$ and two vertices are connected if there is a row with non-zero entries at the corresponding columns.
	The \emph{dual graph of $A$} is defined as $G_D(A) \df G_P(A^\transpose)$, that is, the primal graph of the transpose of $A$.
\end{definition}
\end{comment}

We assume that $G_P(A)$ and $G_D(A)$ are connected, otherwise $A$ has (up to row and column permutations) a block diagonal structure with $d$ blocks and solving~\eqref{MILP} amounts to solving $d$ smaller~\eqref{MILP} instances independently.

\begin{definition}[Treedepth]
	\label{def:tree-depth}
	The {\em closure} $\cl(F)$ of a rooted tree $F$
	is the graph obtained from $F$ by making every vertex adjacent to all of its ancestors.
	%We consider both $F$ and $\cl(F)$ as undirected graphs.
	The \emph{height} of a tree $F$ denoted $\height(F)$ is the maximum number of vertices on any root-leaf path.
	A \emph{$\td$-decomposition of $G$} is a tree $F$ such that $G \subseteq \cl(F)$.
	The {\em treedepth} $\td(G)$ of a connected graph $G$ is the minimum
	height of its $\td$-decomposition.
	%A $\td$-decomposition $F$ of $G$ is \emph{optimal} if $\height(F) = \td(G)$.
\end{definition}
Computing $\td(G)$ is \NPh  but can be done in time $2^{\td(G)^2} \cdot |V(G)|$~\cite{ReidlRVS:2014}, hence \FPT parameterized by $\td(G)$.
To facilitate our proofs we use a parameter called topological height introduced by Eisenbrand et al.~\cite{EisenbrandHunkenschroederKleinKouteckyLevinOnn19}:
\begin{definition}[{Topological height}] \label{def:topheight}
	A vertex of a rooted tree $F$ is \emph{degenerate} if it has exactly one child, and \emph{non-degenerate} otherwise (i.e., if it is a leaf or has at least two children).
	%Note that if the root has only one child then it is degenerate.
	The \emph{topological height of $F$}, denoted $\ttd(F)$, is the maximum number of non-degenerate vertices on any root-leaf path in $F$.
	%Equivalently, $\ttd(F)$ is the height of $F$ after contracting each edge from a degenerate vertex to its unique child.
	Clearly, $\ttd(F) \leq \height(F)$.
	
	%We shall now define the level heights of $F$, which relate to lengths of paths between non-degenerate vertices.
	For a root-leaf path $P=(v_{b(0)}, \dots, v_{b(1)}, \dots, v_{b(2)}, \dots v_{b(e)})$ with $e$ non-degenerate vertices $v_{b(1)}, \dots, v_{b(e)}$ (potentially $v_{b(0)} = v_{b(1)}$), define $k_1(P) \df |\{v_{b(0)}, \dots, v_{b(1)}\}|$, $k_i(P) \df |\{v_{b(i-1)}, \dots, v_{b(i)}\}|-1$ for all $i \in [2,e]$, and $k_i(P) \df 0$ for all $i > e$.
	For each $i \in [\ttd(F)]$, define $k_i(F) \df \max_{P: \text{root-leaf path}} k_i(P)$.
	We call $k_1(F), \dots, k_{\ttd(F)}(F)$ the \emph{level heights of $F$}.
	%See Figure~\ref{fig:td}.
\end{definition}
%\begin{figure}[bt]
%	\centering
%	\begin{subfigure}[b]{0.45\textwidth}
%		\includegraphics[width=\textwidth]{td_nogrey} \hfill
%		\caption{Two optimal $\td$-decompositions $F$ and $F'$ of the cycle on six vertices (in dashed edges). Non-degenerate vertices are enlarged. The trees obtained by contracting edges outgoing from vertices with only one child are pictured below. Notice that even though both $F$ and $F'$ are optimal $\td$-decompositions, their topological height differs. Dashed lines depict ``levels'' of $F$ and $F'$, and we have $k_1(F) = k_2(F) = k_1(F') = 2$ and $k_2(F') = k_3(F') = 1$.} \label{fig:td}
%	\end{subfigure}
%	\hfill
%	\begin{subfigure}[b]{0.53\textwidth}
%		\includegraphics[width=\textwidth]{decomp_smaller_nogrey} \hfill
%		\caption{The situation of Lemma~\ref{lem:decomposition}: an optimal $\td$-decomposition $F$ of $G_P(A)$ pictured in the matrix $A$, the decomposition into smaller blocks $\bar{A}_1, \dots, \bar{A}_d, A_1, \dots, A_d$ derived from $F$ and their $\td$-decompositions $F_1, \dots, F_d$, and a $\td$-decomposition $\hat{F}_d$ of $G_P(\hat{A}_d)$ (Lemma~\ref{lem:fhati}).}
%		\label{fig:decomposition}
%	\end{subfigure}
%	\caption{Illustration of Definitions~\ref{def:tree-depth} and~\ref{def:topheight} (part~\ref{fig:td}) and Lemmas~\ref{lem:decomposition} and~\ref{lem:fhati} (part~\ref{fig:decomposition}).} \label{fig:figs}
%\end{figure}
We also need two lemmas from~\cite{EisenbrandHunkenschroederKleinKouteckyLevinOnn19}.
\begin{lemma}[{Primal Decomposition~\cite[Lemma 19]{EisenbrandHunkenschroederKleinKouteckyLevinOnn19}}] \label{lem:decomposition}
	Let $A \in \Z^{m \times n}$, $G_P(A)$, and a $\td$-decomposition $F$ of $G_P(A)$ be given, where $n,m \geq 1$.
	Then there exists an algorithm computing in time $\Oh(n)$ a decomposition of $A$
	\begin{align}
	A=
	\left(\begin{array}{ccccc}
	\bar{A}_1 & A_1     &  &\\
	\vdots &      &  \ddots & \\
	\bar{A}_d&      &   & A_d
	\end{array}\right), \tag{block-structure} \label{eq:block-structure}
	\end{align}
	and $\td$-decompositions $F_1, \dots, F_d$ of $G_P(A_1), \dots, G_P(A_d)$, respectively,
	where $d \in \N$, $\bar{A}_i \in \Z^{m_i \times k_1(F)}$, $A_i \in \Z^{m_i \times n^i}$, $\ttd(F_i) \leq \ttd(F)-1$, $\height(F_i) \leq \height(F)-k_1(F)$, for $i \in [d]$, $n_1, \dots, n_d, m_1, \dots, m_d \in \N$.
\end{lemma}
%Note that given an~\eqref{MILP}, the primal decomposition naturally partitions the right hand side $\veb = (\veb^1, \dots, \veb^d)$ according to the rows of $A_1, \dots, A_d$, and each object of length $n$ (such as bounds $\vel, \veu$, a solution $\vex$, any step $\veg$, or the objective function $f$) into $d+1$ objects according to the columns of $\bar{A}_1, A_1, \dots, A_d$.
%For example, we write $\vex = (\vex^0, \vex^1, \dots, \vex^d)$.
\begin{lemma}[{\cite[Lemma 21]{EisenbrandHunkenschroederKleinKouteckyLevinOnn19}}] \label{lem:fhati}
	Let $A \in \Z^{m \times n}$, a $\td$-decomposition $F$ of $G_P(A)$, and $\bar{A}_i, A_i, F_i$, for all $i \in [d]$, be as in Lemma~\ref{lem:decomposition}.
	Let $\hat{A}_i \df (\bar{A}_i~A_i)$ and let $\hat{F}_i$ be obtained from $F_i$ by appending a path on $k_1(F)$ new vertices to the root of $F_i$, and the other endpoint of the path is the new root.
	Then $\hat{F}_i$ is a $\td$-decomposition of $\hat{A}_i$, $\ttd(\hat{F}_i) < \ttd(F)$, and $\height(\hat{F}_i) \leq \height(F)$.
\end{lemma}

\section{Fractionality of Bounded-Treedepth Matrices}
Consider any optimal solution $(\vex^\ast_\Z,\vex^\ast_\Q)$ of \eqref{MILP}.
The fractional part $\vex^\ast_\Q$ is necessarily an optimal solution of the \emph{linear} program $\min\{\vecc \vex_\Q:~A_\Q\vex_\Q = \veb - A_\Z \vex^\ast_\Z, \vel_\Q \leq \vex_\Q \leq \veu_\Q, \vex_\Q \in \Q^q\}$. To bound the fractionality of \eqref{MILP}, it therefore suffices to consider the fractionality of $A_\Q$, and we shall hence assume $A=A_\Q$.

Let us now recall some basic facts about vertices of polytopes adapted to the specifics of our situation.
Consider a vertex of the polytope described by the solutions of the system of
\begin{align} \label{eq:general lp}
A\vex = \veb, \vel \leq \vex\leq \veu\,,
\end{align}
with $A,\vex,\vel,\veu$ as usual.
Let $\vex$ be any solution of \eqref{eq:general lp}. 
Being a vertex means satisfying $n$ linearly independent constraints with equality.
Without loss of generality \cite[Proposition 4]{EisenbrandHunkenschroederKleinKouteckyLevinOnn19}, $A$ is \emph{pure}, meaning that its $m$ rows are linearly independent.

Since these first $m$ equations necessarily hold for any solution $\vex$, 
we have $m$ linearly independent constraints satisfied, and there remain $n-m$ of the in total $2n$ upper and lower bounds to be satisfied.
Without loss of generality, we may assume that it is indeed the first $n-m$ lower bound constraints that are met with equality, that is, 
$x_1 = l_1,\ldots,x_{n-m} = l_{n-m}$ holds. 
Let 
\[\vex_N = (x_1,\ldots,x_{n-m}) \in \Q^{n-m}, \vex_B = (x_{n-m+1},\ldots,x_{n}) \in \Q^{m}\,, \]
and partition accordingly the $n$ columns of $A$ as $A = (A_N ~ A_B)$.
Letting $\veb' = \veb - A_N\vex_N$,
the solution $\vex = (\vex_N, \vex_B)$ satisfies
\begin{align}
A_B\vex_B = \veb'.
\end{align}
Observe that $A_B \in \Z^{m\times m}$ is a square matrix with trivial kernel (that is, $A\vex = \ve{0}$ only for $\vex = \ve{0}$), thus invertible.
Therefore, $\vex_B = A_B^{-1} \veb'$.
(Otherwise, there is a direction $\vey$ in the kernel such that both $\vex + \epsilon \vey$ and $\vex - \epsilon \vey$ are feasible, hence $\vex$ was not a vertex.)
Hence, in order to bound the fractionality of the vertex $\vex$, 
it is enough to bound the fractionalities of the entries of $A_B^{-1}$.
We will denote with 
%$\enclen(M)_\infty$ the maximal encoding length of an entry of a matrix $M$, which for a rational number $p/q$ in lowest terms is $O(\log p + \log q)$.
$\fr(A)$ the \emph{fractionality} of $A$, meaning the maximum denominator appearing over all entries, represented as fractions in lowest terms, of $A$.
Note that $\fr(A_1 A_2 \dots A_k) \leq ({\fr(A_1)\cdot\fr(A_2)\cdots \fr(A_k)})!$ for any sequence of matrices $A_1,A_2,\ldots,A_k$,
where we use the aforementioned fact that the least common multiple of numbers bounded by $x$ is bounded by $x!$ .
If one of $A,B$ is $1\times 1$, then $\fr(AB) \leq \lcm(\fr(A),\fr(B))$ holds.

%\label{thm:small_inv}
\begin{reptheorem}{thm:frac}
	Let $A \in \Z^{m \times n}$, and $F$ be a $\td$-decomposition of $G_P(A)$.
	Let $A_B \in \Z^{m\times m}$ be a collection of columns of $A$ forming an invertible submatrix. 
	Then, $\fr(A_B^{-1})$ is bounded by a function of $\height(F)$ and $\|A\|_\infty$.
\end{reptheorem}
Before we proceed with the proof, we will recall the following elementary, but important facts from linear algebra.
A \emph{generalized shear matrix} is a block matrix of the form $M = \begin{pmatrix} I & 0 \\ R & I \end{pmatrix}$, with the blocks being of appropriate size.

\begin{lemma} \label{lem:shearprod}
Let $M_1 = \begin{pmatrix} I & 0 \\ R_1 & I \end{pmatrix}$ be a generalized shear matrix, and $M_2 = \begin{pmatrix} I & 0 \\ R_2 & A \end{pmatrix}$, such that the blocks of $M_1$ and $M_2$ are compatible.
Then, $M_1 \cdot M_2 = \begin{pmatrix} I & 0 \\ R_1 + R_2 & A \end{pmatrix}$.
\end{lemma}
\begin{proof}
Follows directly from the definition of matrix multiplication.
\end{proof}
This is relevant since multiplication with generalized shear matrices corresponds to sequences of elementary row or column transformations.
In particular, we refer to removing the lower left part of a matrix of the form of $M_2$ in the lemma through right or left multiplication with a matrix of the form of $M_1$ with $R_1 = -R_2$ as \emph{zeroing out} $R_2$ from $M_2$.

Let $\operatorname{adj}(A)$ be the matrix having as entries the cofactors of $A$ (commonly called the adjoint of $A$).
For reference, we give Cramer's rule: For invertible $A$,
\begin{align} \label{eq:cramer}
A^{-1} = \frac{1}{\det(A)}\cdot \operatorname{adj}(A)^T\,. \tag{Cramer's rule}
\end{align}
%We now proceed with:
\begin{proof}[Proof of Theorem \ref{thm:frac}]
We induce over $\ttd(F)$.
The base case of $\ttd(F) = 1$ means that $A$ has at most $\height(F) = k_1(F)$ columns, 
and hence by purity also at most $k_1(F)$ rows. 
By \ref{eq:cramer} and the Hadamard bound on determinants, the fractionality of the inverse of any invertible submatrix of $A$, and in particular of $A_B$,  
is therefore bounded by $(k_1(F)\|A\|_\infty)^{k_1(F)}$.

In the induction step, we assume $\ttd(F) > 1$.
Then, $A_B$ inherits a $\td$-decomposition of topological height at most that of $F$, since $G_P(A_B)$ is a subgraph of $G_P(A)$.
We shall therefore assume $A = A_B$ from here on.

Let $r = k_1(F)$.
Consider the block $A_j$ with dimensions $m_j \times n_j$.
Since $A$ is invertible, $m_j \geq n_j$ must hold.
Otherwise, we could combine an all-zeroes column in $A$ from the columns of $A_j$.
Since $A$ is square, $r + \sum_{j=1}^d n_j = \sum_{j=1}^d m_j$
holds.
Letting $r'$ be the number of different values of $j$ such that $m_j > n_j$ is strict,
we see that $r' \leq r$  must hold.
Without loss of generality, we may assume that the first $r'$ inequalities are strict, i.e., $n_1 < m_1,\ldots,n_{r'} < m_{r'}$ holds,
and $n_j = m_j$ for all $j > r'$.

Thus $A$ has the following form (where entries outside of the boxes are zero):
\[
A = \blockmat{$Q_1$}{$R$}{$Q_2$}.\,
\]
%where $Q_1$ has dimensions $left( \sum_{j=1}^{r'} m_j \right) \times \left(r+\sum_{j=1}^{r'} n_j\right),$
%which is square by Eq. \eqref{eq:square_consequence}.
%
%$Q_2$ is consequently of dimension
%$\left( \sum_{j=r'+1}^{r} m_j \right) \times \left( \sum_{j=r'+1}^{r} n_j \right),$
%which is also square.
%
%$R$ is of dimension 
%\[
%\left( \sum_{j=r'+1}^{r} m_j \right) \times r \,.
%\]
Both $Q_1$ and $Q_2$ have to be invertible:
Since row rank equals column rank equals rank, if $Q_1$ did not have full rank, then we could combine one of its rows out of the others. This combination would extend to a combination of rows of the entire matrix $A$. Consequently, $A$ would not be invertible, a contradiction.

Similarly, if $Q_2$ was not invertible, we would get a linear combination of some columns through others, and this would extend to the whole matrix by the same argument, just for columns (or for rows in the transpose).

Let $R' = -Q_2^{-1} \cdot (R ~ \vezero) \cdot Q_1^{-1}$.
Here, $(R~\vezero)$ is the matrix $R$ padded by zero-columns so as to be compatible with $Q_1^{-1}$. Equivalently, we could multiply only with those rows of $Q_1^{-1}$ that correspond to columns of $R$.
By Lemma \ref{lem:shearprod} and elementary matrix calculus,
the inverse of $A$ is given through
\begin{align*}
A^{-1} = \blockmat{$Q_1^{-1}$}{$R'$}{$Q_2^{-1}$}\,.
\end{align*}
Note that the bottleneck term for the fractionality here is $R'$, which contains a product, but we can already say (by the estimate on the fractionality of a product of matrices above) that $\fr(A^{-1})$ depends only on $\fr(Q_1^{-1})$ and $\fr(Q_2^{-1})$, since $R$ is integral.
Therefore, it is enough to bound $\fr(Q_i^{-1})$ individually.
The easier case is $Q_2$, which we take care of now.

By the shape of $A$, $Q_2$ is block diagonal, hence the inverse of $Q_2$ is the block diagonal matrix of the inverses of the blocks $A_i$, $i > r'$.
By Lemma~\ref{lem:decomposition}, each $G_P(A_i)$ has a $\td$-decomposition $F_i$ with $\ttd(F_i) < \ttd(F)$.
We may therefore apply the inductive hypothesis to them, and obtain that 
\begin{align} \label{eq:q2bound}
\fr(Q_2^{-1}) \leq \max_{i}({\fr(A_i)^{-1}}),
\end{align}
and this is bounded.\footnote{Whenever we say some quantity is \emph{bounded}, we mean bounded only by $\|A\|_\infty$ and $\height(F)$, for the remainder of the proof.}
It remains to argue about the inverse of $Q_1$.
Recall that we massaged $A$ such that $Q_1$ contains $r'$ blocks $\hat{A}_i$ for some $r' \leq r$, and the hatted matrices being defined as in Lemma \ref{lem:fhati}.
We now employ another induction, on the number of blocks $Q_1$ is composed of, which is $r'$. 
In particular, we show that the fractionality of $Q_1$ is bounded.

If $r'=0$, this means $A_1$ is empty, and $Q_1$ consists just of one block $\bar{A}_1$ of size $r\times r = k_1(F)\times k_1(F)$.
We can bound the fractionality of $Q_1^{-1}$ again by
$(k_1(F)\|A\|_\infty)^{k_1(F)}$, which is bounded.

Let $\hat{A}_i$ be defined as in Lemma~\ref{lem:fhati}.
Now, assuming $r' > 0$, Lemma~\ref{lem:fhati} crucially states that $G_P(\hat{A}_1)$ has a $\td$-decomposition $\hat{F}_1$ with $\ttd(\hat{F}_1) < \ttd(F)$.
Furthermore, $\hat{A}_1$ is of full rank by purity of $A$.
We may hence pick a set of columns $\hat{B}_1$ of $\hat{A}_1$ that form an invertible submatrix, and by inductive hypothesis \emph{on the topological height, not $r'$}, assume $\fr(\hat{B}_1^{-1})$ is bounded.
We refer to the columns of $\hat{A}_1$ belonging to $\hat{B}_1$ as \emph{invertible}.

Some of the invertible columns may be contained in the first $r$ columns of $\hat{A}_1$, which we call \emph{original}.
By permuting columns, we can bring all the columns of the invertible submatrix $\hat{B}_1$ of $\hat{A}_1$ to the left, having some columns of $\bar{A}_1$ and $A_1$ to its right, to which we now refer to as $N_1$.
That is, $Q_1$ is, up to permutation of columns, of the form 
\begin{align*}
Q_1 = \begin{pmatrix}
\hat{B}_1 & N_1 & \vezero\\
\ast & \ast & \ast
\end{pmatrix}.
\end{align*}
By performing elementary row operations on $\hat{A}_1$ within $Q_1$, we can convert this invertible submatrix $\hat{B}_1$ into an identity matrix, and thus ensure that $Q_1$ has an identity block in the upper left corner.

On the matrix level, this corresponds to left multiplication of $Q_1$ (and after appropriate padding with an identity matrix in the right-bottom corner, also $A$) with a matrix $E_1$ defined as follows:
\begin{align} \label{eq:matop1}
E_1 \cdot Q_1 = 
\begin{pmatrix}
\hat{B}_1^{-1} & \vezero \\
\vezero & I
\end{pmatrix} \cdot \begin{pmatrix}
\hat{B}_1 & N_1 & \vezero\\
\ast & \ast  & \ast
\end{pmatrix} = 
\begin{pmatrix}
I & \hat{B}_1^{-1} N_1 & \vezero\\
\ast & \ast & \ast
\end{pmatrix}.
\end{align}
Note that, below the non-original columns, there are zeroes.
Below the original columns, there are entries of $\bar{A}_i$, $i > 1$, 
which we denote with $\bar{O}$.
Therefore, the right-hand side in Eq. \eqref{eq:matop1} actually reads
$\begin{pmatrix}
I & \hat{B}_1^{-1} N_1 & 0\\
(\vezero ~ \bar{O}) & \ast & \ast
\end{pmatrix}.$
That is, the first asterisk above actually expands to $(\vezero ~ \bar{O})$.
Using Lemma \ref{lem:shearprod} (or rather, a very slight generalization thereof where the top left of $M_2$ is not necessarily an identity matrix), 
we now zero out these entries below the non-original columns, choosing $R_2 = -(\vezero~\bar{O}).$
This corresponds to left-multiplication with 
$E_2 = \begin{pmatrix}
I & \vezero \\
-(\vezero~\bar{O}) & I
\end{pmatrix}$,
yielding a new matrix
\[
\begin{pmatrix}
I & \hat{B}_1^{-1} N_1 & \vezero\\
\vezero & \ast_1 & \ast
\end{pmatrix}.
\]
This modifies the entries below the non-invertible columns of 
(the permuted version of) $\hat{A}_1$, marked $\ast_1$, 
by an additive term of $-(\vezero~\bar{O})\hat{B}_1^{-1}N_1$.

Employing Lemma \ref{lem:shearprod} again, we zero out the non-invertible columns in 
$\hat{A}_1$, that is, $\hat{B}_1^{-1}N_1$,
corresponding to a right multiplication with $E_3 = \begin{pmatrix}
I & \hat{B}_1^{-1}N_1 \\
\vezero & I
\end{pmatrix}$. 
We have thus massaged $Q_1$ into the form 
$\begin{pmatrix}
I & \vezero \\
\vezero & Q_1''
\end{pmatrix}$.
Finally, we need to ensure that the lower-right entry is integral.
To this end, let $\beta = \lcm_{ij}(\fr (Q_1'')_{ij})$, such that $Q_1' \coloneqq \beta \cdot Q_1''$ is integral.
Since $Q_1$ is integral, all fractionality in $Q_1''$ stems from the entries of $\hat{B}_1^{-1}$, which is of dimension $k_1(F)\times k_1(F)$. 
We can hence
bound $\beta \leq \fr (\hat{B}_1^{-1})^{k_1(F)^2}.$
Therefore, we have arrived at the matrix
$1/\beta \cdot \begin{pmatrix}
\beta\cdot I & \vezero \\
\vezero & Q_1'
\end{pmatrix},$
where $Q_1'$ is of the same structure as $Q_1$, that is, the bounds on the topological height still are satisfied, and $k_1(F)$ doesn't change. 
This makes it permissible to induce on $Q_1'$.
Moreover, $Q_1'$ contains one block less than $Q_1$, that is, $r'$ drops by one.
We then apply the inductive hypothesis (with respect to $r'$) to $Q_1'$.

By inductive hypothesis with respect to the topological height, $\fr(\hat{B}_1^{-1})$ is bounded.
By inductive hypothesis on $r'$, 
$\fr((Q_1'')^{-1})$ is bounded by a function only in $\height(F)$ and the size of the entries of $Q_1''$.
The entries of $Q_1'$ in turn are either the intact entries of $A$,
or they were modified by an additive term $-(\vezero~\bar{O})\hat{B}_1^{-1}N_1$.
These are, by inductive hypothesis on $\ttd(F)$, 
also bounded by $\|A\|_\infty$ and $\height(F)$. 

Now, the matrices effecting the transformation given above are $1/\beta,E_1,E_2,E_3$, where we understand the scalar $1/\beta$ as a $1\times 1$ matrix, in the sense that
$\beta\cdot (1/\beta)E_2 E_1 Q_1 E_3 = \begin{pmatrix}
I & \vezero \\
\vezero & Q_1''
\end{pmatrix}$.
Since the result is of block structure, $\begin{pmatrix}
I & \vezero \\
\vezero & (Q_1'')^{-1}
\end{pmatrix} \cdot E_2E_1Q_1E_3 = I$,
implying that $Q_1^{-1} = 1/\beta E_3 \begin{pmatrix}
\beta I & \vezero \\
\vezero & (Q_1')^{-1}
\end{pmatrix} E_2E_1$,
such that $\fr(Q_1^{-1})$ is bounded by a function in $\beta,\fr(E_1),\fr(E_2),\fr(E_3),\fr((Q_1')^{-1})$.
Inspecting the $E_i$ and applying the inductive hypothesis on topological height again, we see that each of these fractionalities depend only on $\|A\|_\infty$ and $\height(F)$. Repeating this step $r' \leq k_1(F)$ times is obviously also within the required bound, since $k_1(F)$ is bounded by $\height(F)$.
Hence, also $\fr Q_1^{-1}$ is bounded. 

Since the entire induction was on $\ttd(F)$, which is bounded by $\height(F)$, the claim holds.
\end{proof}
\begin{corollary}
Let $A, A_B$ be as in Theorem~\ref{thm:frac} and $F$ be a $\td$-decomposition of $G_D(A)$.
Then, $\fr(A_B^{-1})$ is bounded by a function of $\height(F)$ and $\|A\|_\infty$.
\end{corollary}
\begin{proof}
As argued before, $G_D(A_B)$ is a subgraph of $G_D(A)$.
By definition, $G_D(A_B) = G_P(A_B^{\intercal})$.
Finally, $(A_B^{-1})^{\intercal} = (A_B^{\intercal})^{-1}$, and applying the Theorem concludes the proof.
\end{proof}
\begin{remark} \label{rem:nfoldbound}
	The analysis simplifies significantly for the case of 2-stage stochastic and $n$-fold matrices, which are an important special case where $\ttd(F) = 2$: if $t$ is a bound on the block size, then we can bound the fractionality of the submatrix $A_B$ by $\|A\|_\infty^{t^{\Oh(t^2)}}$.
\end{remark}
\begin{proof}[Proof of Theorem~\ref{thm:main}]
Theorem~\ref{thm:frac} gives us a computable bound $M'$ on the largest coefficient of the~\eqref{IMILP} instance, and it is clear that the structure of non-zeroes (hence the primal and dual graphs) of the constraint matrix of~\eqref{IMILP} is identical to that of $A$.
Hence, by Lemma~\ref{lem:IMILP},~\eqref{MILP} can be solved by solving~\eqref{IMILP}, which can be done (by the results of~\cite{EisenbrandHunkenschroederKleinKouteckyLevinOnn19}) in \FPT time parameterized by $\|A\|_\infty$ and $\min\{\td_P(A), \td_D(A)\})$.
(To be precise, we need to solve~\eqref{IMILP} for every $1 \leq \tilde{M} \leq M'$, which is fine as long as $M'$ is computable, and this holds.)
\end{proof}
\begin{proof}[Proof of Corollary~\ref{cor:milp}]
By~\cite[Corollary 93]{EisenbrandHunkenschroederKleinKouteckyLevinOnn19}, $2$-stage stochastic ILP is solvable in time $2^{(2a)^{\Oh(t^3)}} n \log^3 n \log\|\veu-\vel\|_\infty \log \|\vecc\|_\infty$.
The coefficients of~\eqref{IMILP} are bounded by the lcm of numbers of size at most $p := a^{t^{\Oh(t^2)}}$, which is bounded by $a' := p^p$.
The claim follows by plugging in $a = a'$ in the aforementioned runtime.
The same holds for $n$-fold ILP considering that by~\cite[Corollary 91]{EisenbrandHunkenschroederKleinKouteckyLevinOnn19} it is solvable in time $(at^2)^{\Oh(t^3)} n \log n \log \|\veu-\vel\|_\infty \log \|\vecc\|_\infty$.
\end{proof}
\begin{proof}[Proof of Corollary~\ref{cor:lp}]
By~\cite[Corollary 64]{EisenbrandHunkenschroederKleinKouteckyLevinOnn19}, if there is an algorithm solving a class of IP in time $T$, then there is an algorithm solving the corresponding class of LP with accuracy\footnote{To solve an LP with accuracy $\epsilon$ means to find a solution which is at $\ell_\infty$-distance at most $\epsilon$ from an optimum.} $\epsilon$ in time $T \cdot \log(1/\epsilon)$.
By Theorem~\ref{thm:frac}, it is enough to set $1/\epsilon = a^{a^{t^{\Oh(t^2)}}}$, and plugging this into the aforementioned time complexity bounds concludes the proof.
\end{proof}

\section{Hardness Results}
\subsection{High Fractionality Instances}
\begin{proof}[Proof of Lemma~\ref{lem:highfracmilp}]
	It is easy to verify that $A_1^{-1}$ is the inverse of $A_1$ as stated below, both $n \times n$ matrices.
	\begin{align*}
		A = \begin{pmatrix}
			2&-1&0&\cdots&0 \\
			0&2&-1&\cdots&0 \\
			\vdots& &\ddots&&\vdots \\
			0&0&\cdots&&2
		\end{pmatrix}, &&
		A^{-1} = \begin{pmatrix}
			\frac{1}{2}&\frac{1}{2^2}&\frac{1}{2^3}&\cdots&\frac{1}{2^n} \\
			0&\frac{1}{2}&\frac{1}{2^2}&\cdots&\frac{1}{2^{n-1}} \\
			\vdots& &\ddots&&\vdots \\
			0&0&\cdots&&\frac{1}{2}
		\end{pmatrix} \enspace .
	\end{align*}
	Moreover, the primal, dual, and incidence treewidth of $A_1$ is $1$, and $\|A_1\|_\infty = 2$.
	
	It is again easy to verify that below are $A_2$ and its inverse, both $n \times n$, with $n' = n-2$, and $A_2$ is a $4$-block $n$-fold matrix with all blocks of size $1$:
	\begin{align*}
		A = \begin{pmatrix}
			1&1&1& \cdots&1 \\
			1&1&0& \cdots&0 \\
			1&0&1& \cdots&0 \\
			\vdots&&&\ddots & \\
			1&0&0& \cdots&1 \\
		\end{pmatrix}, &&
		A^{-1} = \begin{pmatrix}
			-\frac{1}{n'}&\frac{1}{n'}&\frac{1}{n'}& \cdots&\frac{1}{n'} \\
			\frac{1}{n'}&\frac{n'-1}{n}&-\frac{1}{n'}& \cdots&-\frac{1}{n'} \\
			\frac{1}{n'}&-\frac{1}{n'}&\frac{n'-1}{n}& \cdots&-\frac{1}{n'} \\
			\vdots&&&\ddots & \\
			\frac{1}{n'}&-\frac{1}{n'}&-\frac{1}{n'}& \cdots&\frac{n'-1}{n} \\
		\end{pmatrix} \enspace .
	\end{align*}
	Because for each vertex $\vex$ of a polyhedron there exists an objective vector $\vecc$ such that~\eqref{MILP} is uniquely optimal in $\vex$, and the fact that we have demonstrated inverses with high fractionality, there must exist vertices of high fractionality and corresponding objectives, which give the desired instances $I_1$ and $I_2$.
\end{proof}
\begin{remark}
	The $\Omega(n)$ fractionality lower bound in part 2 of Lemma~\ref{lem:highfracmilp} may be seen as mild given that for $4$-block $n$-fold we would seek an algorithm running in time $n^{f(k)}$, for $f$ some function and $k$ largest block size, and that (the more permissive) $n$-fold IP problem has such an algorithm even when its entries are polynomial in $n$.
	However, this is not true for the $2$-stage stochastic IP problem, which is \NPh with polynomially bounded coefficients already with constant-size blocks~\cite{DvorakEGKO17}.
	Because $4$-block $n$-fold IP is even harder than $2$-stage stochastic IP, the bounded fractionality approach cannot work for giving an $n^{f(k)}$ algorithm for $4$-block $n$-fold MILP.
\end{remark}

\begin{proof}[Proof of Lemma~\ref{lem:highfracmip}]
	All instances have unique optima, and it is straightforward to verify that in part 1 of the Lemma, it is the point $\vex = (\frac{1}{n}, \dots, \frac{1}{n})$, in part 2 it is $x=\frac{1}{k}$, for any $k$, and in part 3, the minimum is irrational $x=\frac{\sqrt{7}}{3} - \frac{2}{3}$, hence fractionality is unbounded.
	The objective $f(x) = x^3 + 2x^2 - x$ is not convex on $\R$, but it is between $0$ and $1$.
\end{proof}

\subsection{The Limits of Tractability for Structured MILPs}
Here, we show hardness Lemma~\ref{lem:hard} about the decision version of MILP, which is deciding the non-emptiness of the following set:
\begin{equation}\label{MILPF}
	\left\{ \vex \in \Z^{z}\times \Q^{q} \mid A\vex = \veb, \, \ve{l} \leq \vex \leq
	\ve{u} \right\} \enspace . \tag{\MILPF{}}
\end{equation}
%\begin{lemma}
%  \MILPF{} is \NPh even if the instance induced by all the integer
%  variables has no constraints and the subinstance induced by all
%  non-integer variables falls within any \NPh fragment of ILP. 
%\end{lemma}
\begin{proof}[Proof of Lemma~\ref{lem:hard}]
	We provide a polynomial-time reduction from ILP-feasibility. Let
	$\III:=\{ \vex \in \Z^n \mid A\vex = \veb, \, \vel \leq \vex\leq \veu\}$
	be an instance of \ILPF{}. Informally, we obtain the
	equivalent instance $\III'$ of MILP by putting the variables of $\III$
	into the non-integer part and then making an (integer) copy of every
	variable in $\III$, which ensures (by forcing the copy to be equal
	to its original) that the original variables can
	only take integer values. More formally, $\III'$ is given by:
	\begin{align}
		\left\{ x' \in \Z^n \times \Q^n \mid \left(\begin{array}{cc}
			& A \\
			I & -I \\
		\end{array}\right) \vex' =
		\left(\begin{array}{c}\veb\\\vezero\end{array}\right), \left(\begin{array}{c}\vel\\\vel\end{array}\right)\leq
		\vex' \leq \left(\begin{array}{c}\veu\\\veu\end{array}\right)\right\}
	\end{align}
	where $I$ is the $n\times n$ identity matrix and $\vezero$ is the
	$n$ dimensional all zero vector. Note that the subinstance induced
	by all integer variables of $\III'$ has no constraints and the
	subinstance induced by all non-integer variables is equal to $\III$.
	(Here, by an \emph{induced} subinstance we mean one obtained by retaining only constraints not containing any of the remaining variables, as those constraints would be arguably meaningless in the induced subinstance.)
\end{proof}
\begin{remark}
	It is an interesting question for future
	work whether we can generalize our results for MILP if we put
	additional restrictions on the interactions between integer and
	non-integer variables. A similar approach has recently been explored
	for generalizing the tractability result for ILP based on primal
	treedepth to MILP~\cite{GanianOrdyniakRamanujan17} using a hybrid
	decompositional parameter called torso-width.
\end{remark}

\bibliographystyle{splncs04}
\bibliography{milp}

\end{document}